\documentclass[runningheads]{llncs}
\usepackage[T1]{fontenc}
\usepackage{graphicx}
\usepackage[dvipsnames]{xcolor}
\usepackage{amsmath,amssymb}
\usepackage{algorithm,algpseudocode}
\usepackage{hyperref}
\usepackage{placeins}
\usepackage{subcaption}
\usepackage[appendix=inline]{apxproof}

\usepackage{color}

\newtheoremrep{theorem}{Theorem}

\date{\today}

\begin{document}
\title{Linear Search with Probabilistic Detection and Variable Speeds}

\author{Jared Coleman\inst{1}\orcidID{0000-0003-1227-2962} \and
Oscar Morales-Ponce\inst{2}\orcidID{0000-0002-9645-1257}}

\authorrunning{J. Coleman and O. Morales-Ponce}

\institute{Loyola Marymount University, Los Angeles, CA 90045, USA \email{jared.coleman@lmu.edu} \and
California State University, Long Beach, Long Beach, CA 90840, USA \email{oscar.moralesponce@csulb.edu}}

\maketitle

\begin{abstract}
    We present results on new variants of the linear search (or cow-path) problem that involves an agent searching for a target with unknown position on the infinite line.
    We consider the variant where the agent can move either at speed $1$ or at a slower speed $v \in [0, 1)$.
    When traveling at the slower speed $v$, the agent is guaranteed to detect the target upon passing through its location.
    When traveling at speed $1$, however, the agent, upon passing through the target's location, detects it with probability $p \in [0, 1]$.
    We present algorithms and provide upper bounds for the competitive ratios for three cases separately: when $p=0$, $v=0$, and when $p,v \in (0,1)$.
    We also prove that the provided algorithm for the $p=0$ case is optimal.
    \keywords{Linear Search \and Cow-Path \and Probabilistic Search \and Mobile Agents \and Competitive Analysis \and Online Algorithms.}
\end{abstract}

\section{Introduction}\label{sec:intro}

The linear search (or cow-path) problem involves designing algorithms for an agent tasked with locating a target of unknown position on an infinite line~\cite{search:beck,search:bellman}. In this paper, we study a new variant of the problem in which the agent operates under two constraints: probabilistic detection and variable movement speeds. This model is motivated by scenarios in robotics and distributed search with trade-offs between speed and detection reliability.
Specifically, we consider an agent, henceforth referred to as the \textit{robot}, that can traverse the line at two speeds: a faster speed of $1$ and a slower speed $v \in [0,1)$. While moving at the slower speed, the robot is guaranteed to detect the target upon passing its location. At the faster speed, detection occurs probabilistically, with probability $p \in [0,1]$ upon passing the target. The goal is to design search strategies that minimize the expected \textit{competitive ratio}, defined as the ratio of the robot's total travel time to the minimum possible time required by an omniscient agent.

This generalization introduces a unique interplay between exploration efficiency and detection reliability, requiring novel approaches to balance the trade-offs. The problem formulation is motivated by real-world applications where search efficiency is constrained by imperfect sensing and mobility limitations. One such application arises in autonomous robotic search-and-rescue missions, where drones or ground robots must locate missing persons in uncertain environments. High-speed traversal allows for rapid coverage of large areas, but the probability of detecting a small target may decrease unless slower, more careful searches are employed. 
Another application is in environmental monitoring and hazardous material detection, where mobile sensors must scan an area for pollutants, leaks, or radioactive sources. Fast movement enables broad sweeps, but sensors often have detection latencies or require sufficient exposure time for reliable readings. Similar challenges arise in underwater exploration, where autonomous underwater vehicles (AUVs) must detect submerged objects with sonar systems that may exhibit probabilistic detection characteristics depending on speed and environmental conditions.
In cybersecurity and computer networks, our model also relates to intrusion detection systems (IDS) scanning for malicious activity. A security algorithm may switch between fast, coarse-grained sweeps of a network and slower, more detailed analyses of suspicious areas, reflecting the trade-offs between speed and detection reliability in anomaly detection.
Beyond these practical applications, our study contributes to the broader theoretical understanding of search and exploration in adversarial and uncertain environments. By characterizing competitive strategies for different parameter settings, we provide insights into the fundamental limits of probabilistic detection in search problems. 

The contributions of this paper are:
\begin{enumerate}
    \item We analyze the problem in the extreme case where $p=0$ (no detection at speed $1$), providing an algorithm with tight competitive bounds.
    \item We consider the other extreme case where $v=0$ (no slow speed), providing an algorithm and its competitive ratio.
    \item For the case where $p \in (0,1)$ and $v \in (0,1)$, we propose a hybrid algorithm.
\end{enumerate}

The rest of the paper is organized as follows. In Section~\ref{sec:related}, we review related work on the linear search problem and its variants. In Section~\ref{sec:fast}, we present an algorithm and its competitive ratio for the case where $v=0$ (i.e., the robot must always move at speed $1$). In Section~\ref{sec:slow}, we present an algorithm and its competitive ratio for the case where $p=0$ (i.e., the robot can only detect the target when moving at speed $v$). In Section~\ref{sec:general}, we discuss the general case where $p,v \in (0,1)$ and propose a hybrid algorithm. Finally, we conclude in Section~\ref{sec:conclusion}.
This is the full version of a paper that will appear in the proceedings of the 36th International Workshop on Combinatorial Algorithms (IWOCA 2025).

\section{Related Work}\label{sec:related}

The linear search problem, also known as the cow-path problem, was first studied by Beck~\cite{search:beck} and Bellman~\cite{search:bellman} in the 1960s. The classical formulation involves an agent searching for a stationary target on an infinite line while minimizing the competitive ratio—the worst-case ratio of the agent's travel distance to the minimum distance required by an omniscient searcher. Over the decades, numerous variations have been investigated, including multi-agent group search and evacuation~\cite{group_search_evac}, moving targets~\cite{moving-target}, faulty agents~\cite{faulty_agents}, and search-and-rescue~\cite{search-and-rescue}. A recent survey~\cite{linear_search_knowledge} discusses the impact of prior knowledge on the competitive ratio in linear search.

A few studies have explored variations of the problem where the agent's speed is non-uniform. In~\cite{beachcomber}, a search model is examined in which an agent can switch between a faster movement mode (without detection capability) and a slower mode that enables detection. This model is closely related to our work when $v=1$ and $p=0$, but differs in that it considers multi-agent search on a ray rather than an infinite line. Another study~\cite{multispeed:terrain-dependent} investigates terrain-dependent speeds, where an agent's velocity depends on its location and direction of travel.

The challenge of uncertain detection has been analyzed in multiple settings. The competitive bound for the $v=0$ case was previously established in~\cite{search_games} (Section 8.6.2), although a matching lower bound remains an open problem. The study in~\cite{pfaulty_ray} examines the case of probabilistic detection on a ray, demonstrating that a non-monotonic strategy is optimal. Similarly,~\cite{multimodal} explores a setting in which an agent can detect a target only when operating in the correct one of multiple searching modes.

To the best of our knowledge, our study is the first to consider a search model that simultaneously incorporates variable speeds and probabilistic detection on an infinite line. 

\section{The Fast-Only Case}\label{sec:fast}
In this section, we consider the case where $v=0$ and so the agent must always move at the high speed $1$.
We start by presenting an algorithm with a competitive ratio of $\frac{8}{p} + \frac{p}{2-p}$, Algorithm~\ref{alg:fast}.
\begin{algorithm}
    \caption{Fast Approach}
    \label{alg:fast}
    \begin{algorithmic}[1]
        \State $a \gets \frac{2}{2-p}$, $i \gets 0$ \Comment{Expansion Ratio, Round}
        \While{target not found}
            \State move to $(-1)^{i+1} a^{i}$ and back to $0$ at speed $1$
            \State $i \gets i+1$ \Comment{Increment Round}
        \EndWhile
    \end{algorithmic}
\end{algorithm}

A version of the following theorem was previously proven in~\cite{search_games} (Section 8.6.2). 
\begin{theoremrep}
    Algorithm~\ref{alg:fast} has an expected competitive ratio of $\frac{8}{p} + \frac{p}{2-p}$.
\end{theoremrep}
\begin{appendixproof}
    Observe that the agent executing Algorithm~\ref{alg:fast} essentially runs in rounds, where each round $i$ consists of the agent moving from the origin to $a^i$ and back to the origin.
    Let $T_i = 2 a^i$ denote the time it takes for the agent to complete round $i$.
    Because the agent moves at speed $1$, it detects the target when it passes through its location with probability $p$.
    Suppose, without loss of generality, that the target is located at position $d \in (a^{i}, a^{i+2}]$ such that it is first passed by the agent during round $i+2$ (whether or not it is detected).
    If the target is detected the first time it is passed (in round $i+2$), then the expected competitive ratio is:
    \begin{align*}
        \frac{\sum_{j=0}^{i+1} 2a^j + d}{d} .
    \end{align*}
    If instead it is detected the \textit{second} time it is passed (in round $i+2$ on the way back toward the origin), the expected competitive ratio is:
    \begin{align*}
        \frac{\sum_{j=0}^{i+1} 2a^j + a^{i+2} + (a^{i+2} - d)}{d}
        = \frac{\sum_{j=0}^{i+2} 2a^j - d}{d} .
    \end{align*}
    We can generalize this, then, and write the expected competitive ratio if the target is detected the $(k+1)$-th time it is passed as:
    \begin{align*}
        \frac{\sum_{j=0}^{i+k+1} 2a^j + (-1)^k d}{d} .
    \end{align*}
    Observe the probability that the target is detected in round $i+k+2$ is $p(1-p)^k$.
    The expected competitive ratio can be written
    \begin{align}
        \sum_{k=0}^\infty &\frac{p(1-p)^k \left( \sum_{j=0}^{i+k-1} 2a^j + (-1)^k d \right)}{d} \nonumber \\
        &= \sum_{k=0}^\infty p(1-p)^k \left( \frac{2 (a^{i+k+2}-1)}{d(a-1)}\right) \label{eq:fast:1} \\
        &= \sum_{k=0}^\infty \left[ \frac{ p(1-p)^k 2 (a^{i+k+2})}{d(a-1)} - \frac{p(1-p)^k}{d(a-1)} \right] \label{eq:fast:2}  \\
        &= \frac{2p a^{i+2}}{{d(a-1)}} \sum_{k=0}^\infty (a(1-p))^k - \frac{p}{d(a-1)} \sum_{k=0}^\infty (1-p)^k \label{eq:fast:3} 
    \end{align}
    where~\eqref{eq:fast:1} follows from the geometric series formula and~\eqref{eq:fast:2} and~\eqref{eq:fast:3} follow from simplifying and removing constants from the series.
    Observe that the remaining series above are geometric series and converge (by the geometric series test) if and only if $a(1-p) < 1$ and $1-p < 1$.
    In other words, the competitive ratio is bounded if and only if $a < \frac{1}{1-p}$.
    Evaluating the above series (via the infinite geometric series formula, $\sum_{k=0}^\infty a r^k = \frac{a}{1-r}$) yields the following expected competitive ratio:
    \begin{align*}
        \frac{1}{d (a-1)} \left( \frac{2a^{i+2}p}{1+a(p-1)} + \frac{p d (a-1)}{2-p} - 2 \right) .
    \end{align*}
    Observe the above is decreasing with respect to $d$.
    Then, since $d > a^{i}$, the expected competitive ratio is maximized as $d \to a^{i}$.
    Thus, the expected competitive ratio is:
    \begin{align*}
        \lim_{d \to a^i} &\frac{1}{d (a-1)} \left( \frac{2a^{i+2}p}{1+a(p-1)} + \frac{p d (a-1)}{2-p} - 2 \right) \\
        &= \frac{2a^2p}{(a-1)(1+a(p-1))} - \frac{2}{a^i(a-1)} + \frac{p}{2-p}
    \end{align*}
    The above is increasing with respect to $i$, so the expected competitive ratio is maximized as $i \to \infty$.
    Thus, the expected competitive ratio is at most:
    \begin{align*}
        \lim_{i \to \infty} &\frac{2a^2p}{(a-1)(1+a(p-1))} - \frac{2}{a^i(a-1)} + \frac{p}{2-p} \\
        &= \frac{2a^2p}{(a-1)(1+a(p-1))} + \frac{p}{2-p} 
    \end{align*}
    Substituting $a = \frac{2}{2-p}$ and simplifying yields the result stated in the theorem.~\qed
\end{appendixproof}

Observe that, as expected, when $p=1$, the competitive ratio is $9$ (consistent with the original linear search problem).

\section{The Slow-Only Case}\label{sec:slow}
In this section, we consider the case where $p=0$ and so the agent cannot detect the target when moving at speed $1$.
For this case, we present an algorithm where the agent always moves at speed $v$, except over sections of the line that have already been explored.
The algorithm is presented in Algorithm~\ref{alg:slow}.

\begin{algorithm}
    \caption{Slow Approach}
    \label{alg:slow}
    \begin{algorithmic}[1]
        \State $a \gets 1 + \sqrt{\frac{2v^2}{v(1+v)}}$, $i \gets 0$ \Comment{Expansion Ratio, Round}
        \While{target not found}
            \State move to $(-1)^{i+1} a^{i-2}$ at speed $1$ \Comment{Traverse previously-explored regions quickly}
            \State move to $(-1)^{i+1} a^{i}$ at speed $v$ \Comment{Traverse new regions slowly}
            \State move back to $0$ at speed $1$
            \State $i \gets i+1$ \Comment{Increment Round}
        \EndWhile
    \end{algorithmic}
\end{algorithm}

\begin{theorem}
    Algorithm~\ref{alg:slow} has a competitive ratio of $3 + 2\sqrt{2+\frac{2}{v}}+\frac{2}{v}$.
\end{theorem}
\begin{proof}
    Observe that the agent executing Algorithm~\ref{alg:slow} essentially runs in rounds, where each round $i$ consists of the agent moving from the origin to $a^i$ and back to the origin.
    Let $T_i$ denote the time it takes for the agent to complete round $i$.
    Observe that $T_0 = 2$ and $T_1 = 2a$, $T_i$ for every $i \geq 2$ can be written:
    \begin{align*}
        T_i = a^{i-2} + \frac{a^i-a^{i-2}}{v} + a^i
    \end{align*}
    Note that, because the agent moves at speed $v$ over unexplored sections, it is guaranteed to detect the target the first time it passes through its location.
    Let us first consider the case where the target is detected in the first or second round. If the target is detected in the first round, then the competitive ratio is clearly $1$ and the theorem holds.
    If the target is detected in the second round, then the target must be located at position $d \in (1, a^2]$.
    The expected competitive ratio is then:
    \begin{align*}
        \frac{2 + d}{d} = 1 + \frac{1}{d}
    \end{align*}
    which is maximized as $d \to 1$.
    Thus, the competitive ratio is at most $2$ and the theorem holds.

    Now let us consider the case where the target is detected in round $i+2$.
    The target, then, must be located at position $d \in (a^{i}, a^{i+2}]$ and the expected competitive ratio is
    \begin{align*}
        \frac{2 + 2a + \sum_{j=2}^{i+1} T_j + a^i + \frac{d-a^i}{v}}{d} .
    \end{align*}
    which, after substituting and evaluating the geometric series, simplifies to:
    \begin{align*}
        \frac{1}{d} \left( \frac{d-a^i}{v}+a^i+\frac{\left(a^{i+2}-1\right) \left(a^2 (1+v)+v-1\right)}{(a-1) a^2 v} \right) .
    \end{align*}
    Observe the above is decreasing with respect to $d$.
    Then, since $d < a^i$, the competitive ratio is maximized as $d \to a^i$.
    Thus, the competitive ratio is at most:
    \begin{align*}
        \lim_{d \to a^i} &\frac{1}{d} \left( \frac{d-a^i}{v}+a^i+\frac{\left(a^{i+2}-1\right) \left(a^2 (1+v)+v-1\right)}{(a-1) a^2 v} \right) \\
        &= 1 + \frac{(a^{i+2}-1)(a^2(1+v)+v-1)}{a^{i+2} (a-1) v}
    \end{align*}
    Since the above is increasing with respect to $i$, the competitive ratio is maximized as $i \to \infty$.
    Thus, the competitive ratio is at most:
    \begin{align*}
        \lim_{i \to \infty} 1 + \frac{(a^{i+2}-1)(a^2(1+v)+v-1)}{a^{i+2} (a-1) v} &= 1 + \frac{a^2 (1+v) + v - 1}{a^2 (a-1) v}
    \end{align*}
    Substituting $a = 1 + \sqrt{\frac{2v^2}{v(1+v)}}$ yields result stated in the theorem.
    ~\qed
\end{proof}

Again, as expected, when $v=1$, the competitive ratio is $9$ (consistent with the original linear search problem).

\begin{theorem}
    If $p=0$, there exists no algorithm with a competitive ratio better than $3 + 2\sqrt{2+\frac{2}{v}}+\frac{2}{v}$.
\end{theorem}
\begin{proof}
    Let $l(t) \geq 0$ and $r(t) \geq 0$ denote the leftmost and rightmost points such that the entire interval $[l(t), r(t)]$ has been explored (traversed by the agent at the slow speed $v$ at least once) at time $t$.
    We will refer to the interval $[l(t), r(t)]$ as the \textit{explored region}, $l(t)$ as the \textit{left boundary} of the explored region, and $r(t)$ as the \textit{right boundary} of the explored region.

    \begin{figure}
        \centering
        \includegraphics[width=0.85\textwidth]{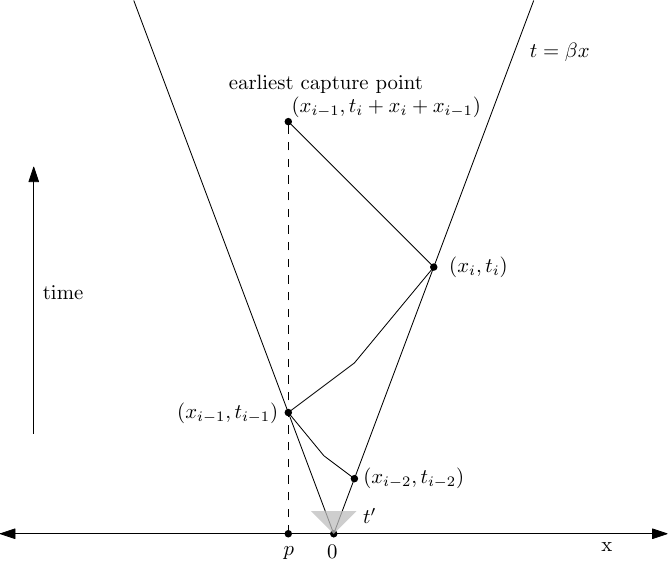}
        \caption{Cone of exploration for the agent.}
        \label{fig:slow_lower_bound}
    \end{figure}

    Then, let $\beta_t = \inf_{t^\prime > t} \frac{t^\prime}{|l(t^\prime)| + |r(t^\prime)|}$. Clearly if $t_1 < t_2$ then $\beta_{t_1} \geq \beta_{t_2}$. Furthermore, $\beta_t \geq 1/v$ for all $t$ since the maximum speed at which the agent can explore a region is $v$.
    Now let $\beta = \lim_{t \to \infty} \beta_t$.
    By the definition of the limit infimum, there must exist a finite time $t$ such that for all $\beta \leq \frac{t_1}{|l(t_1)| + |r(t_1)|}$ for all $t^\prime > t$ and there must exist infinitely many $t^\prime$ such that the (without loss of generality) the right boundary of the explored region $r(t^\prime) = \frac{t^\prime}{\beta}$ (touches the the line $t = \beta x$).

    Consider such a time $t_i$.
    Then, we can define $x_j$ for $j < i$ as the furthest point from the origin the agent could have traveled on the opposite side of the line at time $t_{j+1}$ and $t_j$ as the minimum time the agent could have reached $x_j$ ($t_j = \beta x_j$ if $x_j$ is on the right side of the line and $t_j = -\beta x_j$ if $x_j$ is on the left side of the line).
    See Figure~\ref{fig:slow_lower_bound} for an illustration.

    Observe, then, that $t_i$ can be bounded recursively:
    \begin{align*}
        t_i &\geq t_{i-1} + \frac{|x_i| - |x_{i-1}|}{v} + |x_i| + |x_{i-1}| \\
        t_i &\geq t_{i-1} + \frac{1}{\beta} \left( \frac{t_i - t_{i-1}}{v} + t_i + t_{i-1} \right)
    \end{align*}

    Using trivial base-case bounds of $t_0 \geq 0$ and $t_1 \geq 1$, this recurrence can be solved\footnote{We solved this recurrence using Mathematica. The notebook code is open-source and available at \url{https://github.com/kubishi/code_multi_speed_cow}.} to show that:
    \begin{align*}
        t_i \geq \frac{\beta v - 1}{x} 
        \left(
            \left(
                \frac{v(1+\beta)+x}{2 (\beta  v-1)}
            \right)^i-
            \left(
                \frac{v(1+\beta)-x}{2(\beta  v-1)}
            \right)^i
        \right)
    \end{align*}
    where $x = \sqrt{4+v(6 \beta v + (1 + \beta^2) v - 4 (1+\beta))}$.
    
    Consider the case where the target is located just to the left of $x_{i-1}$ so that it could not have been detected before time $t_i$.
    Then the earliest time the target could have been detected is at time $t_{i} + |x_i| + |x_{i-1}| + \epsilon$ for some arbitrarily small $\epsilon > 0$.
    Then, the competitive ratio can be written:
    \begin{align}
        \text{CR} &= \sup_{\epsilon,i} \frac{t_i + |x_i| + |x_{i-1}| + \epsilon}{|x_{i-1}| + \epsilon} = \sup_{\epsilon,i} \left[ 1 + \frac{t_i + |x_i|}{|x_{i-1}| + \epsilon} \right] \nonumber \\ 
        &= \sup_{i} \left[ 1 + \frac{t_i + |x_i|}{|x_{i-1}|} \right] 
        = \sup_{i} \left[ 1 + (1+\beta) \frac{t_i}{t_{i-1}} \right] \nonumber \\
        &= \sup_{i} \left[ 1 + \frac{(1+\beta) (v(1+\beta)+x)}{2 (\beta  v-1) \left(
            1+\frac{2 x}{v(1+\beta)-x - (v(1 + \beta)+x)^i (v(1 + \beta)-x)^{1-i}}
        \right)
        } \right] \label{eq:lower_bound:1} \\
        &= \sup_{i} \left[ 1 + \frac{(1+\beta) (v(1+\beta)+x)}{2 (\beta  v-1) \left(
            1+\frac{2x}{(v(1+\beta)-x) \left( 1 - \left(\frac{v(1 + \beta)+x}{v(1 + \beta)-x}\right)^i \right)} 
        \right) 
        } \right] \label{eq:lower_bound:2} .
    \end{align}
    Equation~\ref{eq:lower_bound:1} results from substituting $t_i$ and $t_{i-1}$ into the competitive ratio.
    Equation~\ref{eq:lower_bound:2} follows since $v(1 + \beta)-x$ (and therefore $v(1 + \beta)+x$) is always positive.
    Observe that, since $\frac{v(1 + \beta)+x}{v(1 + \beta)-x} > 1$, \eqref{eq:lower_bound:2} is decreasing with respect to $i$.
    Thus, the competitive ratio is maximized as $i \to \infty$.
    Then, the competitive ratio is at most
    \begin{align*}
        1 + \frac{(1+\beta) (v(1+\beta)+x)}{2 (\beta  v-1) \left(
            1+\frac{2x}{(v(1+\beta)-x)} 
        \right) 
        } .
    \end{align*}
    Substituting $x = \sqrt{4+v(6 \beta v + (1 + \beta^2) v - 4 (1+\beta))}$ yields
    \begin{align*}
        1 + \frac{(1+\beta) \left(v(1+\beta) + \sqrt{4+v(6 \beta v + (1 + \beta^2) v - 4 (1+\beta))}\right)}{2 (\beta v-1)} .
    \end{align*}
    Using standard minimization techniques, it can be verified that the minimum is achieved with $\beta = 1 + 2/v$.
    Substituting this value of $\beta$ yields the result stated in the theorem.~\qed
\end{proof}

\section{The General Case}\label{sec:general}
Clearly, if both $v=0$ and $p=0$, the agent will never find the target and if both $v=1$ and $p=1$ the problem is equivalent to the original cow-path problem (and  Algorithms~\ref{alg:fast} and~\ref{alg:slow} are equivalent and optimal).
The only interesting case left, then, is when $v \in (0,1)$ and $p \in (0,1)$.
In this case, the agent must decide when to move at speed $1$ and when to move at speed $v$.

\subsection{Moving Fast or Slow}
The first thing to note about this case is that both Algorithm~\ref{alg:fast} and~\ref{alg:slow} (and their competitive ratios) remain valid for this model.
A first question, then, is to determine for what values of $p$ and $v$ either algorithm is better than the other.
Let $CR_{Fast}$ and $CR_{Slow}$ denote the expected competitive ratios of Algorithm~\ref{alg:fast} and Algorithm~\ref{alg:slow}, respectively.
Then, we have the following result:
\begin{align}
    CR_{Fast} &< CR_{Slow} \nonumber \\
    3 + 2\sqrt{2+\frac{2}{v}}+\frac{2}{v} &< \frac{8}{p} + \frac{p}{2-p} \nonumber \\
    v &< \frac{2p}{8 - p (1-p)^2 - 2 \sqrt{p(8 + p^2 (2-p))}}
    \label{eq:fast_slow}
\end{align}
Thus, a simple algorithm with a competitive ratio of $\max\{CR_{Fast}, CR_{Slow}\}$ is to run Algorithm~\ref{alg:fast} if Inequality~\eqref{eq:fast_slow} holds and Algorithm~\ref{alg:slow} otherwise.
Figure~\ref{fig:regions} shows the values of $p$ and $v$ for which Algorithm~\ref{alg:fast} is better/worse than Algorithm~\ref{alg:slow}.

\subsection{A Hybrid Approach}
In this section, we introduce a hybrid search strategy that balances the trade-offs between the slow and fast search approaches. The algorithm follows the same structure as the slow search algorithm but incorporates an additional "scouting" step. After reaching the turnaround point in each round, instead of immediately returning to the origin, the agent briefly continues forward at the fast speed before heading back.

The intuition behind this approach is that it disrupts the worst-case scenario of the slow search, where the target is placed just beyond the turnaround point. In the standard slow search algorithm, the agent would only reach such a target in a later round, leading to a high competitive ratio. However, by scouting ahead at high speed, the hybrid algorithm has a chance to detect the target much earlier, potentially improving performance in these cases. The challenge is to carefully balance the length of the scouting phase so that it provides meaningful improvement without excessively increasing the total search time for each round.

\begin{algorithm}
    \caption{Hybrid Approach}
    \label{alg:hybrid}
    \begin{algorithmic}[1]
        \State \textbf{inputs}: expansion ratio $a$ and scout-ahead ratio $b$
        \State $i \gets 0$ \Comment{Round}
        \While{target not found}
            \State move to $(-1)^{i+1} a^{i-2}$ at speed $1$ \Comment{Traverse previously-explored regions quickly}
            \State move to $(-1)^{i+1} a^{i}$ at speed $v$ \Comment{Traverse new regions slowly}
            \State move to $(-1)^{i+1} \left( a^{i} + b (a^{i+2} - a^i) \right)$ at speed $1$ \Comment{Scout ahead at the fast speed}
            \State move back to $0$ at speed $1$
            \State $i \gets i+1$ \Comment{Increment Round}
        \EndWhile
    \end{algorithmic}
\end{algorithm}

First, observe that Algorithm~\ref{alg:hybrid} is equivalent to Algorithm~\ref{alg:slow} if $b \gets 0$.
Thus, the expected competitive ratio of Algorithm~\ref{alg:hybrid} is at most the expected competitive ratio of Algorithm~\ref{alg:slow}.
Unfortunately, the expected competitive ratio of Algorithm~\ref{alg:hybrid} is difficult to compute.

\begin{theorem}
    The expected competitive ratio of Algorithm~\ref{alg:hybrid} with expansion ratio $a$ and scout-ahead ratio $b$ is $\max \{ CR_1, CR_2 \}$
    where $CR_1$ is
    \begin{align*}
        \frac{1}{\left(a^2-1\right) b+1} \left[ \frac{\left(a^2-1\right) b}{v}+\frac{a^2 \left(a^2 \left(2 \left(a^2-1\right) b v+v+1\right)+v-1\right)}{(a-1) a^2v}+1 \right] 
    \end{align*}
    and $CR_2$ is
    \begin{align*}
        (a+1) &\left[
            2 a^2 b (p-1)^2+a \left(\frac{1}{a-1}-2 b (p-1) p\right) \right. \\
            &\hspace{1em}\left. +\frac{(p-2) p (v-1)}{a v}+p (2 b+p-2)+\frac{(p-1)^2}{v} 
        \right]  .
    \end{align*}
\end{theorem}
\begin{proof}
    Suppose the target is located at position $d \in (a^{i}, a^{i+2}]$.
    Then there are a few cases to consider:
    \begin{enumerate}
        \item[1:] The target is $d \geq a^{i} + b (a^{i+2} - a^i)$. In this case, the target is detected during round $i+2$.
        \item[2:] The target is $d < a^{i} + b (a^{i+2} - a^i)$. In this case, the target \text{may} be detected during round $i$.
        \begin{itemize}
            \item[2.a:] The target is detected during round $i$ while the agent is ``scouting ahead''.
            \item[2.b:] The target is detected during round $i+2$.
        \end{itemize}
    \end{enumerate}
    Let $T_i$ denote the time it takes for the agent to complete round $i$:
    \begin{align*}
        T_i = a^{i-2} + \frac{a^i-a^{i-2}}{v} + 2 b (a^{i+2} - a^i) + a^i
    \end{align*}
    For case 1, the expected competitive ratio is:
    \begin{align}\label{eq:case1}
        \sup_{i, d \in (a^{i} + b (a^{i+2} - a^i), a^{i+2}]}
        \frac{\sum_{j=0}^{i+1} T_j + a^i + \frac{d-a^i}{v}}{d}
    \end{align}
    Observe that~\eqref{eq:case1} is decreasing with respect to $d$.
    Then, since $d < a^i + b (a^{i+2} - a^i)$, the competitive ratio is maximized as $d \to a^i + b (a^{i+2} - a^i)$.
    The competitive ratio, then, can be written:
    \begin{align*}
        \sup_{i} \frac{\frac{\left(a^2-1\right) b}{v}+\frac{\left(a^2-a^{-i}\right) \left(a^2 \left(2 \left(a^2-1\right) b v+v+1\right)+v-1\right)}{(a-1) a^2v}+1}{\left(a^2-1\right) b+1}
    \end{align*}
    Since the above is increasing with respect to $i$, the competitive ratio is maximized as $i \to \infty$.
    Thus, the competitive ratio is at most:
    \begin{align}
        \frac{1}{\left(a^2-1\right) b+1} \left[ \frac{\left(a^2-1\right) b}{v}+\frac{a^2 \left(a^2 \left(2 \left(a^2-1\right) b v+v+1\right)+v-1\right)}{(a-1) a^2v}+1 \right] \label{eq:hybrid:case1}
    \end{align}

    For case 2, the agent detects the target during round $i$ while scouting ahead on the way to $a^{i} + b (a^{i+2} - a^i)$ with probability $p$ and on the way back to the origin with probability $p (1-p)$.
    With probability $(1-p)^2$, the agent does not detect the target while scouting ahead and detects it only during round $i+2$.
    The expected competitive ratio is:
    \begin{align}
        \sup_{i, d \in (a^{i}, a^{i} + b (a^{i+2} - a^i))}
        &\left[ p \left( \sum_{j=0}^{i-1} T_j + a^i + \frac{d-a^i}{v} \right) / d \right. + \nonumber \\
        &\phantom{x}\left. p (1-p) \left( \sum_{j=0}^{i+1} T_j + a^{i+2} + \frac{d-a^{i+2}}{v} \right) / d  \right. \label{eq:case2} \\
        &\phantom{x}\left. (1-p)^2 \left( \sum_{j=0}^{i+1} T_j + a^{i+2} + \frac{d-a^{i+2}}{v} \right) / d  \right] \nonumber
    \end{align}
    Observe again that~\eqref{eq:case2} is decreasing with respect to $d$.
    Then, since $d < a^{i} + b (a^{i+2} - a^i)$, the competitive ratio is maximized as $d \to a^{i}$.
    Then,~\eqref{eq:case2} can be simplified to:
    \begin{align}\label{eq:case2_simplified}
        \frac{1}{a^2 (a-1) v} & \Bigg[
            \frac{1 - v - a^2 (1 + v + 2 (a-1)^2 b v)}{a^i} + \\
            & \quad a (1+a) \bigg(
                (a-1) \big[a (1-p)^2 + (2-p) p\big] + \nonumber \\
            & \quad \quad \big(
                a^2 + 2 (a-1) a^3 b + 2 p 
                - 2 a (a+b+a^2 (2a-3) b) p \nonumber \\
            & \quad \quad \quad + (a-1) (1+a+2 (a-1) a^2 b) p^2 
            \big) v \bigg) \nonumber 
        \Bigg] 
    \end{align}
    Observe now that~\eqref{eq:case2_simplified} is decreasing with respect to $i$ since 
    \begin{align*}
        1 - v - a^2 (1 + v + 2 (a-1)^2 b v) < 1 - v - a^2 < 0
    \end{align*}
    when $a > 1, b \in [0,1], v \in (0,1)$.
    Thus, the competitive ratio is maximized as $i \to \infty$.
    \begin{align}
        (a+1) &\left[
            2 a^2 b (p-1)^2+a \left(\frac{1}{a-1}-2 b (p-1) p\right) \right. \nonumber \\
            &\hspace{1em}\left. +\frac{(p-2) p (v-1)}{a v}+p (2 b+p-2)+\frac{(p-1)^2}{v}
        \right]  \label{eq:hybrid:case2}
    \end{align}
    The expected competitive ratio is then the maximum of~\eqref{eq:hybrid:case1} and~\eqref{eq:hybrid:case2}.
    ~\qed
\end{proof}

Unfortunately, minimizing the expected competitive ratio of the Hybrid Approach analytically with respect to both $a$ and $b$ has proven difficult.
Figure~\ref{fig:AB} shows the numerically optimized values of the expansion ratio $a$ and the scout-ahead ratio $b$ for different combinations of $p$ and $v$. As expected, the value of $b$ increases with $p$, reflecting the increased benefit of scouting ahead when detection at high speed is more reliable. The expansion ratio $a$ also varies with both parameters, reflecting the trade-off between coverage and redundancy in uncertain detection regimes.
\begin{figure}[!h]
    \centering
    
    \begin{subfigure}{0.5\textwidth}
        \centering
        \includegraphics[width=1.1\textwidth]{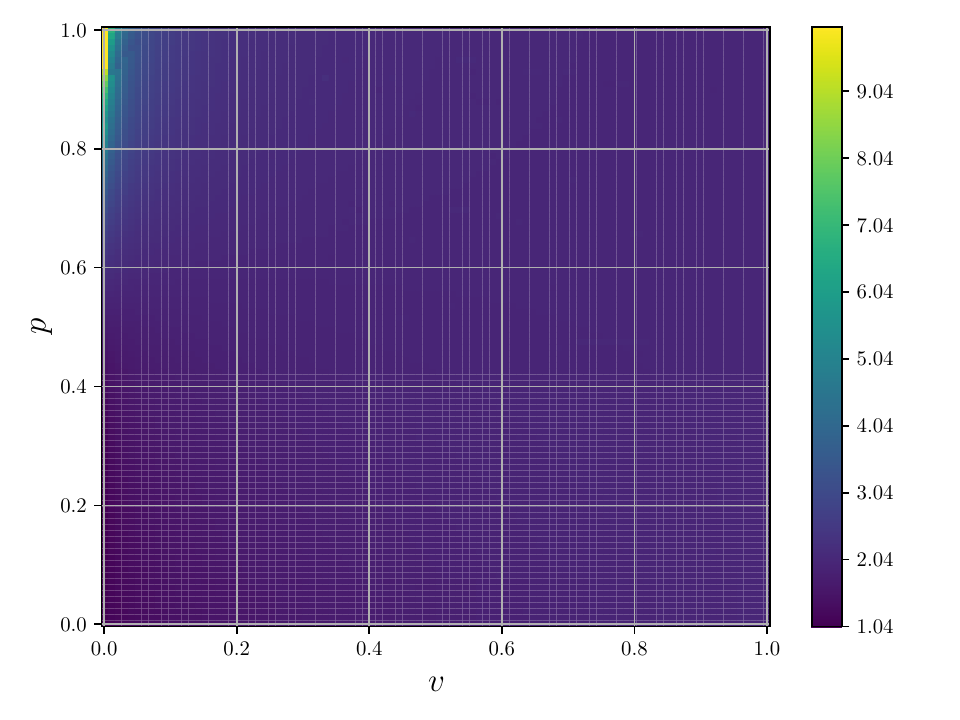}
    \end{subfigure}%
    \begin{subfigure}{0.5\textwidth}
        \centering
        \includegraphics[width=1.1\textwidth]{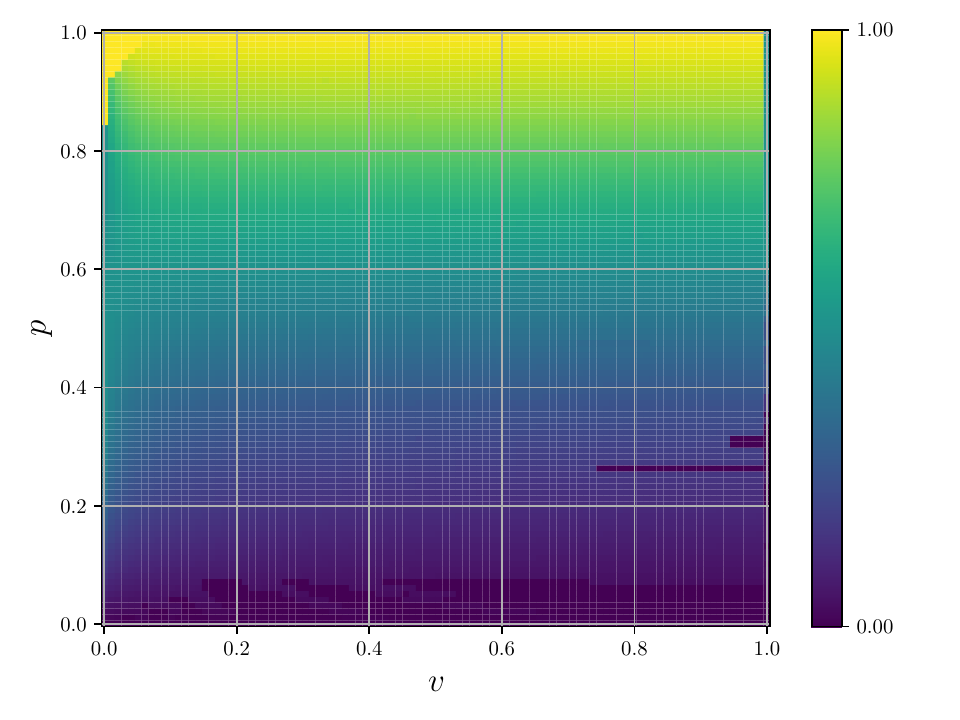}
    \end{subfigure}
    
    \caption{The numerically computed best expansion ratio $a$ (left) and scout-ahead ratio $b$ (right) values for different values of $p$ and $v$.}
    \label{fig:AB}
    
\end{figure}
Figure~\ref{fig:CR} illustrates the expected competitive ratio of the best-performing algorithm between the Fast Approach and the Hybrid Approach\footnote{The code used to generate all figures is open-source and available at \url{https://github.com/kubishi/code_multi_speed_cow}.}. Figure~\ref{fig:CR} shows the improvement of the Hybrid Approach over the Slow Approach, demonstrating its advantage. Finally, Figure~\ref{fig:regions} shows the regions in the $p$-$v$ plane where the Slow Approach and the Hybrid Approach outperform the Fast Approach.
The Hybrid Approach never performs worse than the Slow Approach and expands the parameter space in which it surpasses the Fast Approach. This highlights the effectiveness of a hybrid strategy that balances exploration speed with targeted scouting.
\begin{figure}[!h]
    \centering
    
    \begin{subfigure}{0.5\textwidth}
        \centering
        \includegraphics[width=1.05\textwidth]{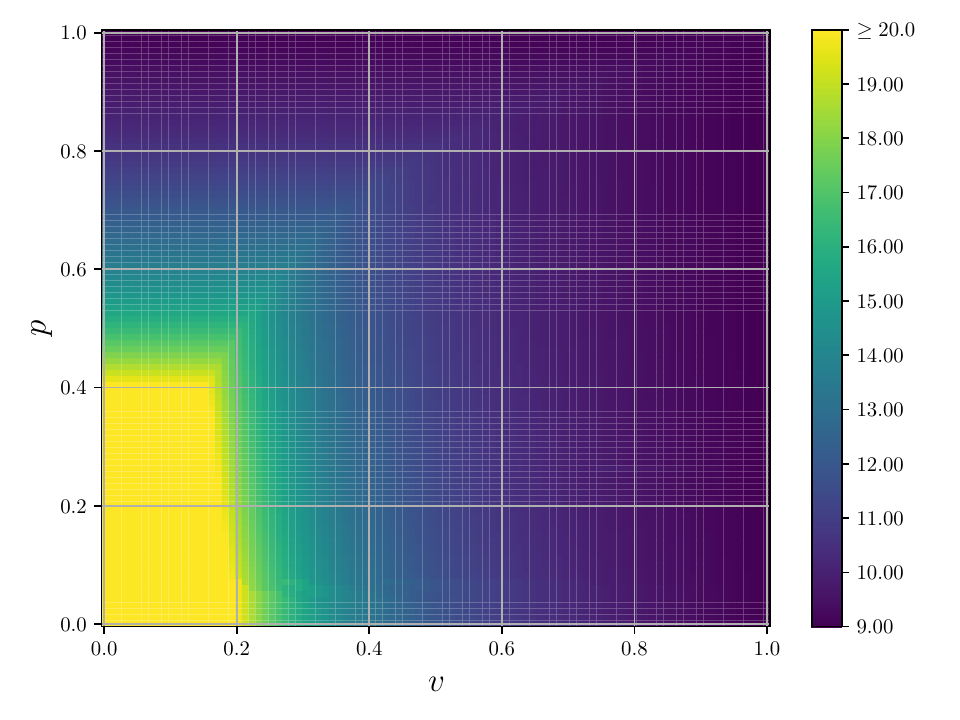}
    \end{subfigure}%
    \begin{subfigure}{0.5\textwidth}
        \centering
        \includegraphics[width=1.05\textwidth]{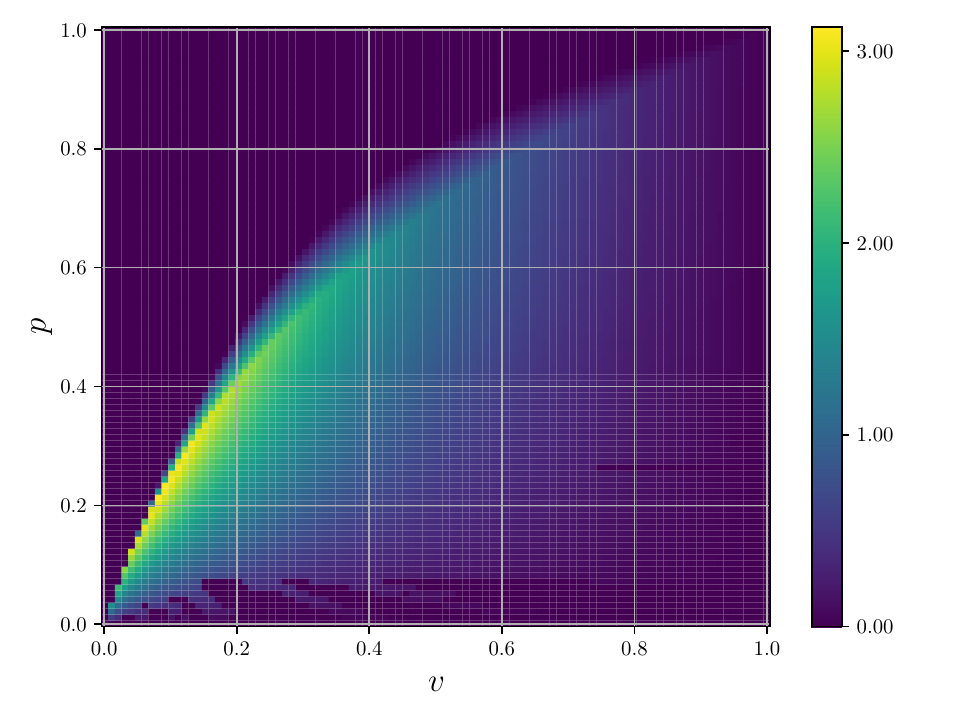}
    \end{subfigure}
    
    \caption{Competitive ratio (left) and difference in expected competitive ratios (right) of Algorithm~\ref{alg:slow} and Algorithm~\ref{alg:hybrid}, showing how much the hybrid approach improves on the slow approach.}
    \label{fig:CR}
    
\end{figure}
\begin{figure}[!h]
    \centering
    \includegraphics[width=\textwidth]{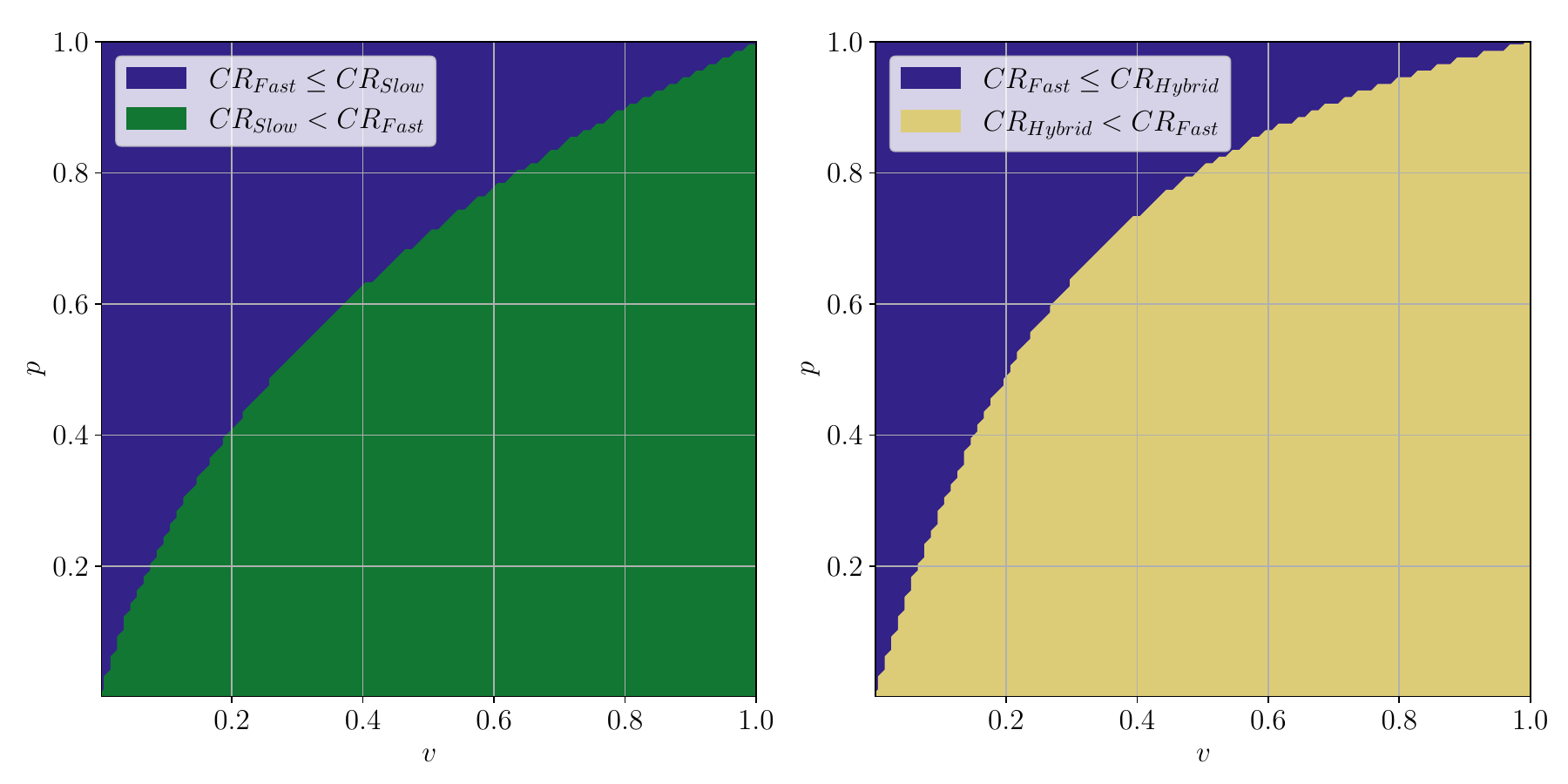}
    \caption{Values for $p$ and $v$ where the Slow and Hybrid Approaches are better/worse than the Fast Approach.}
    \label{fig:regions}
\end{figure}

\section{Conclusion}\label{sec:conclusion}
In this paper, we introduced and analyzed a novel variant of the linear search (cow-path) problem, where a searching agent balances probabilistic detection with variable movement speeds. We presented competitive algorithms for three key cases: (1) when the agent cannot detect the target at high speed ($p=0$), (2) when the agent cannot move at a slow speed ($v=0$), and (3) the general case where both $p$ and $v$ are within $(0,1)$.

For the $p=0$ case, we derived an optimal algorithm with a competitive ratio of $3 + 2\sqrt{2+\frac{2}{v}}+\frac{2}{v}$. For the $v=0$ case, we established an algorithm with a competitive ratio of $\frac{8}{p} + \frac{p}{2-p}$. In the general case, we demonstrated that a hybrid algorithm leveraging both slow and fast movement strategies improves upon existing methods.

Several directions remain for future work. First, tightening the lower bound for the $v=0$ case remains an open question. Additionally, it would be interesting to explore models where the detection probability depends continuously on speed (i.e., $p(v)$), or where speed and sensing capabilities degrade over time. Extensions to higher-dimensional settings or bounded domains (e.g., finite intervals or graphs) may reveal different trade-offs. Finally, considering online variants where $p$ or $v$ is initially unknown and must be learned adaptively could connect this problem to learning-augmented algorithms and online optimization.

\bibliographystyle{splncs04}
\bibliography{refs}

\end{document}